\documentclass[12pt]{amsart}

\usepackage{enumerate}
\usepackage{amsmath}
\usepackage{amssymb}
\usepackage{amsthm}
\usepackage{booktabs}
\usepackage{caption}
\usepackage{graphicx}
\usepackage{hyperref}
\usepackage{url}

\theoremstyle{plain}
\newtheorem{theorem}{Theorem}[section]

\theoremstyle{definition}

\theoremstyle{remark}
\newtheorem*{remark}{Remark}

\newcommand{\pto}{\overset{P}{\to}}
\newcommand{\dto}{\overset{\mathcal{D}}{\to}}

\begin{document}

\title{Measurable Taylor's Theorem: An Elementary Proof}
\author{Gianluca Viggiano*}
\date{May 5, 2023}

\thanks{Keywords: Taylor Expansion, Measurable Function, Asymptotic Analysis, Delta Method}
\thanks{MSC2020 subject classifications: 28A20, 60A05}

\thanks{* Bank of Italy, Regional Economic Research, Milan. E-mail: \href{mailto:gianluca.viggiano@bancaditalia.it}{gianluca.viggiano@bancaditalia.it}}

\maketitle

\begin{abstract}
The Taylor expansion is a widely used and powerful tool in all branches of Mathematics, both pure and applied. In Probability and Mathematical Statistics, however, a stronger version of Taylor's classical theorem is often needed, but only tacitly assumed. In this note, we provide an elementary proof of this \emph{measurable} Taylor's theorem, which guarantees that the interpolating point in the Lagrange form of the remainder can be chosen to depend measurably on the independent variable.
\end{abstract}

\section{Introduction}
\label{sec:intro}
In Real Analysis Taylor's theorem with the Lagrange form of the remainder states that if $f: \mathbb{R} \to \mathbb{R}$ is differentiable $k+1$ times, then it can be globally approximated by its Taylor polynomial at any point, namely for any $x, c \in \mathbb{R}$ there exists a point $\xi$ lying between $c$ and $x$ such that
\begin{equation}
\label{eq:RealTaylor}
 f(x) = T^k_c(x) + \frac{f^{(k+1)}(\xi)}{(k+1)!}(x-c)^{k+1} 
\end{equation}
where $T^k_c(x) = \sum_{j=0}^k \frac{f^{(j)}(c)}{j!}(x-c)^j$. The theorem extends easily to $\mathbb{R}^N$, assuming $f$ is \emph{continuously} differentiable $k+1$ times (see Theorem~\ref{thm:RTTlagrange} below). Of course, the interpolating point $\xi$ need not be unique, and depends on all parameters $c, k, f$ and $x$. 
For fixed $c, k$ and $f$, $x \mapsto \Xi(x)  \equiv \{\xi \text{ such that } \eqref{eq:RealTaylor} \text{ holds}\}$ is then a non-empty correspondence. Using the Axiom of Choice, we can extract a \emph{selection} from $\Xi$, \emph{i.e.}, a function $\xi$ such that $\xi(x) \in \Xi(x)$ for all $x$, but nothing can be said \emph{a priori} about any such function's local or global properties. In many applications, however, some degree of regularity is needed. In Probability Theory, in particular, the existence of a \emph{measurable} selection is often needed and usually assumed without proof.

An illuminating example is the \emph{folklore} proof of the delta method in Mathematical Statistics (see for example \cite[p.~88]{jiang2010large}): suppose $(\mathbf{X}_n)_{n \geq 1}$ is a sequence of random $N-$vectors such that $\sqrt{n}(\mathbf{X}_n - \mathbf{c}) \dto \mathbf{X}$ for some $\mathbf{c} \in \mathbb{R}^N$ and random $N$-vector $\mathbf{X}$. If $f: \mathbb{R}^N \to \mathbb{R}$ is a continuously differentiable function, we can study the asymptotic behavior of $(f(\mathbf{X}_n))_{\geq 1}$ with a first-order Taylor expansion:
\begin{equation}\label{eq:deltam}
\sqrt{n} [f(\mathbf{X}_n) - f(\mathbf{c})] = \nabla f(\boldsymbol{\xi}_n)^T \sqrt{n} (\mathbf{X}_n - \mathbf{c})
\end{equation}
where $\boldsymbol{\xi}_n$ is a point on the segment between $\mathbf{c}$ and $\mathbf{X}_n$. If we assume a random sequence of such $\boldsymbol{\xi}_n$ exists, then it is easy to prove that $\nabla f(\boldsymbol{\xi}_n) \pto \nabla f(\mathbf{c})$ and conclude that $\sqrt{n} [f(\mathbf{X}_n) - f(\mathbf{c})]) \dto \nabla g(\mathbf{c})^T \mathbf{X}$.
This however does not follow in any way from Taylor's classical theorem. In fact, if the $\mathbf{X}_n$ are defined on some sample space $(\Omega, \mathcal{F})$, what Taylor's theorem actually guarantees is that for each $n\in \mathbb{N}$ and $\omega \in \Omega$ we can find a $\boldsymbol{\xi}(n, \omega) \equiv \boldsymbol{\xi}(\mathbf{X}_n(\omega), f, \mathbf{c})$ satisfying~\eqref{eq:deltam}. There is no assurance that the map $\omega \mapsto\boldsymbol{\xi}(n, \omega)$ is actually a random variable, \emph{i.e.} that it is measurable.

For an example from Probability, we can look at the classical proof of It\={o}'s lemma of \cite[p.~171]{karatzas1991brownian}. If $f: \mathbb{R} \to \mathbb{R}$ is of class $C^2$ and $(X_t)_{t \geq 0}$ is a continuous semimartingale, the first step of the proof is to choose a partition $0 = t_0 < t_1 < \cdots < t_N = t$ and use a second-order Taylor expansion to write:
\begin{multline*}
f(X_t) - f(X_0) = \sum_{n=1}^N \{ f(X_{t_n}) - f(X_{t_{n-1}})\} \\
=\sum_{n=1}^N f'(X_{t_{n-1}})(X_{t_n} - X_{t_{n-1}}) +
\frac{1}{2} \sum_{n=1}^N f''(\xi_n) (X_{t_n} - X_{t_{n-1}})^2
\end{multline*}
for some $\xi_n$ lying between $X_{t_{n-1}}$ and $X_{t_n}$. Again, the authors assume that each $\xi_n$ is a random variable as in subsequent steps they use bounds involving its second moments. Notice that in this case we are making the even stronger assumption that the Lagrange interpolator depends measurably on both $c$ and $x$.

To the best of our knowledge \cite[p. 604]{charalambos2013infinite} is the only work to provide a measurable extension of Taylor's theorem. Using advanced results from the theory of measurable selections, the authors prove what they call a \emph{Stochastic Taylor's Theorem} for real functions defined on bounded intervals---essentially a special case of Theorem~\ref{thm:MTTlagrange} below.
The proof appeals to non-trivial results like Kuratowski–Ryll-Nardzewski Selection Theorem and Filippov's Implict Function Theorem, making it accessible only to a very specialized readership. \cite{yang2021note} extends the result to the bi-dimensional case, but the approach does not generalize naturally to arbitrary dimension and still leans on Filippov's Theorem.

We think that providing a clear statement and a reasonably simple proof of this measurable Taylor's theorem can encourage its use at both  intermediate and advanced levels, improving proofs' readability and rigor. In the following, we provide a direct proof of the measurability of the Lagrange interpolator that presupposes only intermediate knowledge of Measure Theory.

\section{Measurable Taylor's Theorem}
\label{sec:MTT}
We will prove the theorem directly in its multivariate form. Below we recall the general statement of Taylor's theorem using multi-index notation (for details see \cite{lang2001real} and \cite{loom2014calc}).

\begin{theorem}[Taylor]
\label{thm:RTTlagrange}
Let $k \in \mathbb{N}$, $U \subseteq \mathbf{R}^N$ open, $f:U \to \mathbb{R}$ a function of class $C^{k+1}$, and $\mathbf{c} \in U$. If $\mathbf{x} \in U$ is such that the segment of endpoints $\mathbf{x}$ and $\mathbf{c}$ is contained in $U$, then there exists a $\boldsymbol{\xi}$ lying on this segment such that the following \emph{Taylor formula with the Lagrange form of the remainder} holds:

\begin{equation}
\label{eq:RTTlagrange}
 f(\mathbf{x}) = T_{\mathbf{c}}^k f(\mathbf{x}) + 
 \sum_{|\alpha| = k+1} \frac{1}{\alpha!} D^{\alpha}f(\boldsymbol{\xi})(\mathbf{x}-\mathbf{c})^{\alpha}
\end{equation}
where  $T_{\mathbf{c}}^k f (\mathbf{x}) = \sum_{|\alpha| \leq k} \frac{1}{\alpha!} D^{\alpha} f(\mathbf{c})(\mathbf{x} - \mathbf{c})^{\alpha}$.
\end{theorem}

Adopting the same notation, we can now state the following:

\begin{theorem}[Measurable Taylor's Theorem]
\label{thm:MTTlagrange}
Let $k \in \mathbb{N}$, $U \subseteq \mathbb{R}^N$ open and star-convex with respect to $\mathbf{c} \in U$, and $f: U \to \mathbb{R}$ of class $C^{k+1}$. Then there exists a Borel-measurable function $\boldsymbol{\xi}: U \to U$ such that, for all $\mathbf{x} \in U$, $\boldsymbol{\xi}(\mathbf{x})$ lies on the segment of endpoints $\mathbf{x}$ and $\mathbf{c}$, and satisfies

\begin{equation}
\label{eq:lagXi}
 f(\mathbf{x}) = T_{\mathbf{c}}^k f(\mathbf{x}) + 
  \sum_{|\alpha| = k+1} \frac{1}{\alpha!} D^{\alpha}f(\boldsymbol{\xi}(\mathbf{x}))(\mathbf{x}-\mathbf{c})^{\alpha}
\end{equation}
Moreover, $\boldsymbol{\xi}$ can be chosen to be continuous at $\mathbf{c}$.
\end{theorem}

\begin{proof}
For $\mathbf{x} \in U$ and $t \in [0,1]$, define
\[
F(\mathbf{x}, t) = f(\mathbf{x}) - T_{\mathbf{c}}^k f(\mathbf{x}) - \\
\sum_{|\alpha| = k+1} \frac{1}{\alpha!} D^{\alpha}f(\mathbf{c} + t(\mathbf{x} - \mathbf{c}))(\mathbf{x}-\mathbf{c})^{\alpha}.
\]
With this notation, Taylor's theorem can be restated as guaranteeing that for all $\mathbf{x} \in U$ the equation $F(\mathbf{x}, t)=0$ has at least one solution in $t \in [0,1]$. Define

\begin{equation}
\label{eq:tau}
\tau(\mathbf{x}) = \inf \{ t \in [0, 1] : F(\mathbf{x}, t) = 0 \}.
\end{equation}
By continuity of $F$, $\tau$ verifies $F(\mathbf{x}, \tau(\mathbf{x})) = 0$. We claim that $\tau$ is also measurable: to prove this we will construct a sequence of measurable functions $\tau_n$ converging pointwise to $\tau$.
Define
\begin{equation}
\label{eq:taun}
\tau_n(\mathbf{x}) = \inf \{ q \in \mathbb{Q} \cap [0, 1] : |F(\mathbf{x}, q)| < 1/n \}.
\end{equation}
$\tau_n$ is measurable since, for $t \leq 0$ and $t > 1$, $\{ \tau_n < t\}$ is equal to the empty set and $U$, respectively, while for $t \in (0,1]$ we have

\begin{equation}\label{eq:semicont}
\{ \tau_n < t \} =
\bigcup_{q \in \mathbb{Q} \cap [0, t)} \{
\mathbf{x} : |F(\mathbf{x},q)| < 1/n \},
\end{equation}

In addition, $n \mapsto \tau_n(\mathbf{x})$ is non-decreasing in $n$ for all $\mathbf{x} \in U$, so the limit function $\overline{\tau}(\mathbf{x}) = \lim_{n \to \infty} \tau_n(\mathbf{x})$ exists everywhere, is measurable and verifies $F(\mathbf{x}, \overline{\tau}(\mathbf{x})) = 0$ by continuity of $F$. 

We claim that $\overline{\tau} = \tau$: Fix $\mathbf{x} \in U$. Since $\overline{\tau}(\mathbf{x})$ belongs to the right-hand side of \eqref{eq:tau}, we have $\tau(\mathbf{x}) \leq \overline{\tau}(\mathbf{x})$. Vice versa, take $t \in [0, 1]$ such that $F(\mathbf{x}, t) = 0$. If $t$ is rational then it belongs to the right-hand side of \eqref{eq:taun} for all $n \in \mathbb{N}$; if it is irrational we can still find, for each $n \in \mathbb{N}$, a $q_n \in \mathbb{Q} \cap [0,1]$ such that $q_n < t$ and $|F(\mathbf{x}, q_n)| < 1/n$. In both cases we have $\tau_n(\mathbf{x}) \leq t$ and letting $n \to \infty$ gives us $\overline{\tau}(\mathbf{x}) \leq t$. Since $t$ was arbitrary in the right-hand side of \eqref{eq:tau}, we conclude  that $\overline{\tau}(\mathbf{x}) \leq \tau(\mathbf{x})$.

Finally, simply define $\boldsymbol{\xi}(\mathbf{x}) = \mathbf{c} + \tau(\mathbf{x}) (\mathbf{x} - \mathbf{c})$. The only non-trivial statement left to verify is continuity at $\mathbf{c}$. Notice that $\boldsymbol{\xi}(\mathbf{c}) =\mathbf{c}$, since $\tau (\mathbf{c}) = 0$ by construction. Continuity then follows immediately from the fact that $|\boldsymbol{\xi}(\mathbf{x}) - \mathbf{c}| \leq |\mathbf{x} - \mathbf{c}|$, again by construction.
\end{proof}

\begin{remark}
A careful reading of the previous proof shows that we actually never used the fact that $\mathbf{c}$ was fixed. In fact, we could rewrite each step, letting all functions defined depend on both $\mathbf{x}$ and $\mathbf{c}$, and all statements would remain valid. Thus, we have actually proved a slightly stronger result, which comes in handy in situations like the proof of It\={o}'s lemma from the Introduction:
\end{remark}

\begin{theorem}[Measurable Taylor's Theorem, symmetric version]
\label{thm:MTTlagrangeSymm}
Let $k \in \mathbb{N}$, $U \subseteq \mathbb{R}^N$ open and convex, and $f: U \to \mathbb{R}$ of class $C^{k+1}$. Then there exists a Borel-measurable function $\boldsymbol{\xi}: U \times U \to U$ such that for all $\mathbf{x}, \mathbf{y} \in U$, $\boldsymbol{\xi}(\mathbf{x}, \mathbf{y})$ lies on the segment of endpoints $\mathbf{x}$ and $\mathbf{y}$ and satisfies

\begin{equation}
\label{eq:lagXiSymm}
 f(\mathbf{x}) = T_{\mathbf{y}}^k f(\mathbf{x}) + 
  \sum_{|\alpha| = k+1} \frac{1}{\alpha!} D^{\alpha}f(\boldsymbol{\xi}(\mathbf{x}, \mathbf{y}))(\mathbf{x}-\mathbf{y})^{\alpha}
\end{equation}
\end{theorem}

\section{Example}
\begin{figure}
\centering
\includegraphics[scale=0.7]{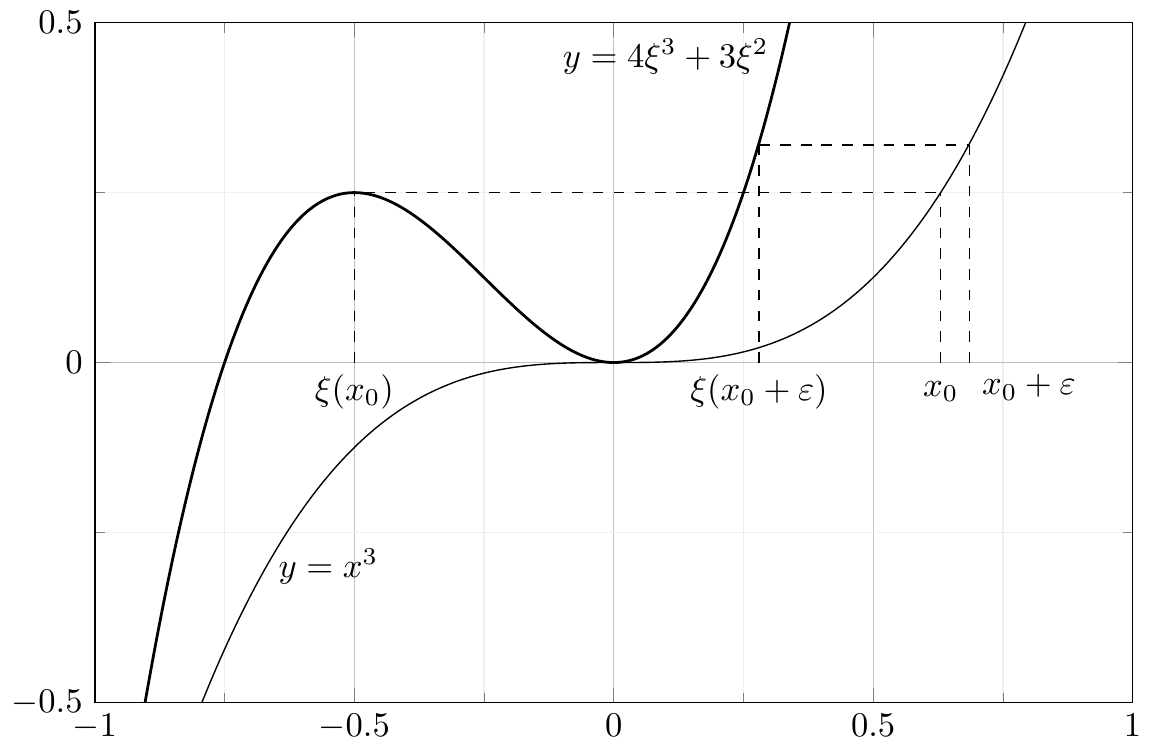}
\caption{Example of a discontinuous  $\xi(x)$}
\label{fig:counterex}
\end{figure}

\label{sec:EXample}
In general, the function we constructed is not guaranteed to be continuous at points different from $\mathbf{c}$. For a counterexample consider the case $k=1$, $U = \mathbb{R}$, $f(x) = x^3(x+1)$ and $c = -1$. Then for $x > -1$, $\xi(x)$ is defined as the least solution in $[-1, x]$ to the equation
\begin{equation}\label{eq:exmeastt}
  f'(\xi) = \frac{f(x) - f(-1)}{x+1},
\end{equation}
that is to say
\begin{equation}
\label{eq:counterex}
  4\xi^3 + 3\xi^2  =x^3.
\end{equation}
A glance at the graphs of the two polynomials in Figure~\ref{fig:counterex} shows that $\xi(x)$ must have a jump at $x_0 = 1/ \sqrt[3]{4}$ where the right-hand side of \eqref{eq:counterex} is equal to $1/4$. In fact, for values slightly larger than $x_0$, say $x_0 + \varepsilon$, we have only one (positive) $\xi$ satisfying \eqref{eq:counterex}, while for $x_0$ we also have a negative solution, so that $\xi(x_0) < 0 < \xi(x_0^+)$.

\begin{remark}
Looking at the previous example and proof one might be tempted to assume that, at least in the one-dimensional case, $\xi$ is semicontinuous. In fact, the approximating functions $\tau_n$ are upper semicontinuous by Equation~\ref{eq:semicont}, but unfortunately, they \emph{increase} to $\tau$, whereas upper semicontinuity is preserved only by \emph{decreasing} limits.

More generally, the previous example shows that $\xi$ cannot be guaranteed to be upper semicontinuous, exhibiting a case where $\limsup_{x \to x_0} \xi(x) = \xi(x_0^+) > 0 >  \xi(x_0)$. It is not difficult to construct an example where $\liminf_{x \to x_0} \xi(x) = \xi(x_0^-) < 0 <  \xi(x_0)$, thus showing that neither lower semicontinuity holds in general: it is enough to reflect the previous example across the $y$-axis, \emph{i.e.}, take $g(x) = f(-x)$, $c = 1$ and $x_0 = -1/ \sqrt[3]{4}$. In fact, in this case, for $x < 1$, $\xi(x)$ will be the largest solution in $[x, 1]$ to the analog of Equation~\ref{eq:exmeastt} and $\xi$ will be right-continuous at $x_0$ but not left-continuous.
\end{remark}


\begin{thebibliography}{6}

\bibitem{charalambos2013infinite} C. D. Aliprantis and K. C. Border, \emph{Infinite Dimensional Analysis: A Hitchhiker’s Guide}, 3rd ed., Springer, Berlin, 2006.

\bibitem{lang2001real} S. Lang, \emph{Real and Functional Analysis}, 3rd ed., Grad. Texts Math., Springer, Berlin, 1993.

\bibitem{karatzas1991brownian} I. Karatzas and S. Shreve \emph{Brownian motion and stochastic calculus}, 2nd ed., Grad. Texts Math., Springer, Berlin, 1991.

\bibitem{loom2014calc} H. L. Lynn and S. Shlomo, \emph{Advanced Calculus}, Rev. ed., World Scientific, Singapore, 2014.

\bibitem{jiang2010large} Jang J., \emph{Large Sample Techniques for Statistics}, Stat. Texts Stat., Springer, New York, 2010.

\bibitem{yang2021note} Y. Yang and X. Zhou, \emph{A Note on Taylor’s Expansion and Mean Value Theorem With Respect to a Random Variable}, arXiv:2102.10429 (2021). 
\end{thebibliography}
\end{document}